\documentclass[11pt,a4paper]{article}

\usepackage{graphicx,amsfonts,amsbsy, amssymb, amsthm, amsmath, mathrsfs}

\newcommand{\rmd}{{\mathrm d}}
\newcommand{\rme}{{\mathrm e}}
\newcommand{\rmi}{{\mathrm i}}

\newcommand{\Ord}{{\mathrm O}}
\newcommand{\littleo}{{\mathrm o}}

\DeclareSymbolFont{lettersA}{U}{pxmia}{m}{it}

\DeclareMathSymbol{\alphaup}{\mathord}{lettersA}{"0B}
\DeclareMathSymbol{\betaup}{\mathord}{lettersA}{"0C}
\DeclareMathSymbol{\gammaup}{\mathord}{lettersA}{"0D}
\DeclareMathSymbol{\deltaup}{\mathord}{lettersA}{"0E}
\DeclareMathSymbol{\epsilonup}{\mathord}{lettersA}{"22}
\DeclareMathSymbol{\zetaup}{\mathord}{lettersA}{"10}
\DeclareMathSymbol{\etaup}{\mathord}{lettersA}{"11}
\DeclareMathSymbol{\thetaup}{\mathord}{lettersA}{"12}
\DeclareMathSymbol{\iotaup}{\mathord}{lettersA}{"13}
\DeclareMathSymbol{\kappaup}{\mathord}{lettersA}{"14}
\DeclareMathSymbol{\lambdaup}{\mathord}{lettersA}{"15}
\DeclareMathSymbol{\muup}{\mathord}{lettersA}{"16}
\DeclareMathSymbol{\nuup}{\mathord}{lettersA}{"17}
\DeclareMathSymbol{\xiup}{\mathord}{lettersA}{"18}
\DeclareMathSymbol{\piup}{\mathord}{lettersA}{"19}
\DeclareMathSymbol{\rhoup}{\mathord}{lettersA}{"1A}
\DeclareMathSymbol{\sigmaup}{\mathord}{lettersA}{"1B}
\DeclareMathSymbol{\tauup}{\mathord}{lettersA}{"1C}
\DeclareMathSymbol{\upsilonup}{\mathord}{lettersA}{"1D}
\DeclareMathSymbol{\phiup}{\mathord}{lettersA}{"1E}
\DeclareMathSymbol{\chiup}{\mathord}{lettersA}{"1F}
\DeclareMathSymbol{\psiup}{\mathord}{lettersA}{"20}
\DeclareMathSymbol{\omegaup}{\mathord}{lettersA}{"21}

\renewcommand{\Psi}{\varPsi}
\renewcommand{\Lambda}{\varLambda}
\renewcommand{\Sigma}{\varSigma}
\renewcommand{\Gamma}{\varGamma}
\renewcommand{\Theta}{\varTheta}
\renewcommand{\Xi}{\varXi}
\renewcommand{\Pi}{\varPi}
\renewcommand{\Upsilon}{\varUpsilon}
\renewcommand{\Phi}{\varPhi}
\renewcommand{\Omega}{\varOmega}

\newcommand{\R}{{\mathbb R}}
\newcommand{\N}{{\mathbb N}}

\newcommand{\Z}{{\mathbb Z}}

\newcommand{\C}{{\mathbb C}}

\newcommand{\Prob}{{\mathbb P}}
\newcommand{\E}{{\mathbb E}}

\newcommand{\coloneq}{\mathbin{\hbox{\raise0.08ex\hbox{\rm :}}\!\!=}}
\newcommand{\eqcolon}{\mathbin{=\!\!\hbox{\raise0.08ex\hbox{\rm :}}}}

\renewcommand{\leq}{\leqslant}
\renewcommand{\geq}{\geqslant}
\renewcommand{\epsilon}{\varepsilon} 

\newcommand{\dimostrazione}{\noindent{\sl Proof.}\phantom{X}}

\newcommand{\finire}{\hfill$\Box$}

\newcommand \printdate[3]{%
    \def \@suffix##1{%
        \def \@n{##1}%
        \ifnum \@n = 1 st\else%
        \ifnum \@n = 2 nd\else%
        \ifnum \@n = 3 rd\else%
        \ifnum \@n = 21 st\else%
        \ifnum \@n = 22 nd\else%
        \ifnum \@n = 23 rd\else%
        \ifnum \@n = 31 st\else%
        th\fi \fi \fi \fi \fi \fi \fi%
    }%
    \relax%
    \number #1\raise0.7ex\hbox{\footnotesize \@suffix{#1}}\kern0.25em%
    \ifcase #2\or%
        January\or February\or March\or%
        April\or May\or June\or%
        July\or August\or September\or%
        October\or November\or December%
    \fi\ %
    \number #3%
}

\newtheorem{theorem}{Theorem}[section]
\newtheorem{proposition}[theorem]{Proposition}
\newtheorem{definition}[theorem]{Definition}
\newtheorem{corollary}[theorem]{Corollary}
\newtheorem{lemma}[theorem]{Lemma}
\newtheorem{remark}[theorem]{Remark}
\newtheorem{assumption}[theorem]{Assumption}

\newcommand{\Ea}{E_{\rm a}}
\newcommand{\Ec}{E_{\rm c}}
\newcommand{\Eb}{E_{\rm b}}
\newcommand{\Ed}{E_{\rm d}}

\newcommand{\curlyD}{{\mathcal D}}
\newcommand{\dom}{\mathop{\rm Dom}}
\newcommand{\asterisk}{*} 

\newcommand{\curlyP}{{\mathcal P}}
\newcommand{\curlyM}{{\mathcal M}}
\newcommand{\qqev}{\mu} 
\newcommand{\qqef}{\psi_{\sigma,I}}
\newcommand{\hatqqef}{\hat{\psi}_{\sigma,I}}
\newcommand{\qqqef}[1]{\psi_{\sigma,I,#1}}

\begin{document}
\title{Localised eigenfunctions in \v{S}eba billiards}
\author{J.P.~Keating${}^1$ \and J.~Marklof${}^1$ \and B.\ Winn${}^{2}$\\
{\protect\small\em ${}^1$ School of Mathematics,
University of Bristol, Bristol, BS8 1TW, U.K.}\\
{\protect\small\em ${}^2$ School of Mathematics, 
Loughborough University, Loughborough,
LE11 3TU, U.K. }}
\date{\printdate{18}{3}{2010}}
\maketitle
\begin{abstract}
We describe some new families of quasimodes for the Laplacian
perturbed by the addition of a potential formally described by
a Dirac delta function. As an application we find, under some
additional hypotheses on the spectrum, 
subsequences of eigenfunctions of \v{S}eba
billiards that localise around a pair of unperturbed eigenfunctions.
\end{abstract}
 
\thispagestyle{empty}

\section{Introduction}
One of the unsolved questions in the analysis of quantum eigenfunctions
concerns possible limiting distributions as the eigenvalue tends to infinity.
For eigenfunctions of the Laplace operator on certain surfaces with
arithmetical properties,
it has been proved \cite{lin:ima, sou:que} that there is only one possible
limit; all sequences of eigenfunctions become uniformly distributed. On the
other hand, Hassell \cite{has:ebt} has proved the existence of chaotic
billiard domains in $\R^2$ for which zero-density subsequences of eigenfunctions
fail to equi-distribute in the limit. 

We consider the Laplace operator plus potential supported at a single point.
Such a potential has been 
variously referred to as delta-interaction potential, 
Fermi pseudo-potential or zero-range potential in different parts of the
literature. Mathematically this operator can be constructed using the
tools of self-adjoint extension theory.

We will prove our results for the case where the underlying space is a
compact 2-dimensional manifold, for which the Laplace operator has
eigenfunctions and eigenvalues denoted by $\Phi_j$ and $E_j$ respectively.
We perturb this operator with a delta potential supported at the point 
$p$, which will remain fixed throughout, and suppressed from
notations. This perturbation
can be realised by a 1-parameter family of self-adjoint operators,
$H_\Theta$, indexed by an angle $\Theta$ which controls the strength of
the perturbation.

We fix a finite interval $I\subseteq\R$ containing at least one $E_j$, and
define, for notational convenience,
\begin{equation}
\zeta_I(s,\lambda)\coloneq \sum_{E_j\in I}\frac{|\Phi_j(p)|^2}
{(E_j-\lambda)^s}.
\end{equation}
Let $\sigma\in[0,1]$. We define
\begin{equation} \label{eq:intro:qqef}
\psi(x)\coloneq \sum_{E_j\in I}\frac{\overline{\Phi_j(p)}}{E_j-\qqev}
\Phi_j(x) + \sigma \sum_{E_j\not\in I}\left( E_j-
\frac{\sin\Theta}{1-\cos\Theta}\right)\frac{\overline{\Phi_j(p)}}{1+E_j^2}
\Phi_j(x),
\end{equation}
for $\qqev$ a solution to
\begin{equation}  \label{eq:intro:qqev}
\zeta_I(1,\mu)=\sigma \sum_{E_j\in I}\left( E_j-
\frac{\sin\Theta}{1-\cos\Theta}\right)\frac{|\Phi_j(p)|^2}{1+E_j^2}.
\end{equation}

Our main results are as follows:
\begin{theorem} \label{thm:quasi_discrep}
The pair $(\psi,\qqev)$ is a quasimode for $H_\Theta$ with
discrepancy $d$, where
\begin{equation}
d^2 = \frac{(1-\sigma)^2 \zeta_I(0,\mu) 
 + \sigma^2\sum_{E_j\not\in I} \left( 1+E_j\qqev +\frac{\sin\Theta}{1-\cos\Theta}
(E_j-\qqev)\right)^2\frac{|\Phi_j(p)|^2}{(1+E_j^2)^2}}
{\zeta_I(2,\mu) + \sigma^2 \sum_{E_j\not\in I}
\left(E_j-\frac{\sin\Theta}{1-\cos\Theta}\right)^2\frac{|\phi_j(p)|^2}
{(1+E_j^2)^2}}
\end{equation}
Furthermore, if $\psi_1, \psi_2$ are defined by \eqref{eq:intro:qqef} for
$\qqev_1\neq \qqev_2$, two solutions of \eqref{eq:intro:qqev} then
  \begin{equation}
    \left\langle \psi_1, \psi_2 \right\rangle =
\sigma^2 \sum_{E_j\not\in I} \left( E_j-\frac{\sin\Theta}{1-\cos\Theta}\right)^2
\frac{|\Phi_j(p)|^2}{(1+E_j^2)^2}.
  \end{equation}
\end{theorem}

The construction of families of quasimodes is a key step in Hassell's
proof \cite{has:ebt}, as well as the proofs of many recent results
on localisation of quantum eigenfunctions \cite{bog:swf2, bur:bbm, col:maq,
don:que, hil:coe, mar:ql}. One reason for this
is that quasimodes can often be used to approximate eigenfunctions.
In general (see the introduction to section \ref{sec:drei} for precise
statements) the smaller the discrepancy, the closer quasimodes are
to true eigenfunctions. For this reason 
it is important to know when the discrepancy
can be made small. In this direction we have the following corollary
to theorem \ref{thm:quasi_discrep}:


\begin{corollary} \label{cor:arbitrary:small}
Let $\sigma=1$ and let $I=[0,T]$ where $T>E_1$. Then the discrepancy $d$
of the quasimode $\psi$ satisfies
\begin{equation}
d\ll \frac{\qqev}{\sqrt{T}}.
\end{equation}

Let $\sigma=0$ and $I$ be any interval containing at least two $E_j$.
If $\mu\in I$ then the discrepancy $d$ of $\psi$ satisfies
\begin{equation}d\leq \frac1{\sqrt{2}}\ell(I),\end{equation}
 where $\ell(I)$ is the length of $I$.
If, additionally, $I$ contains precisely two $E_j$, then we 
have \begin{equation} d\leq \frac12\ell(I). \end{equation}
\end{corollary}

In particular, the quasimodes with $\sigma=1$ and $\mu$ held fixed or
slowly growing, can be made arbitrarily
precise by choosing $T$ as large as desired.

We are interested in ascertaining when true eigenfunctions of $H_\Theta$
have mass supported on our quasimodes. Without any assumptions on
the spectrum of the Laplacian we can prove the following.

\begin{proposition} \label{prop:sqrt3}
For any consecutive eigenvalues $\Ea<\Eb<\Ec<\Ed$ from the sequence
$(E_j)_{j=1}^\infty$, let $I=[\Eb,\Ec]$ and take $\sigma=0$. Choose $\mu$ so
that $\mu\in I$. Then
there is an eigenfunction $\phi$ of $H_\Theta$ with eigenvalue
in the interval $(\Ea,\Ed)$ such that
\begin{equation} \label{eq:intro:overlap}
|\langle \phi,\psi\rangle| \geq\frac{\|\psi\|}{\sqrt{3}}\left(1-
\frac{(\Ec-\Eb)^2} 
{4\min\{\Ed-\Ec,\Eb-\Ea\}^2}\right)^{1/2}.
\end{equation}
\end{proposition}

Proposition \ref{prop:sqrt3} is most interesting when the sequence
of eigenfunctions $\Phi_j$ do not equi-distribute. (For example if
they are solutions to a PDE which is subject to separation of variables;
see below.) Then, by considering an infinite subset of the 
spectrum $\{E_j\}$ along which the right-hand side of 
\eqref{eq:intro:overlap} is bounded away from zero, proposition
\ref{prop:sqrt3} proves the existence of a sequence of eigenfunctions
of $H_\Theta$ which fail to equi-distribute. Such a subset of $\{E_j\}$ does
exist since the mean level spacing is constant.

Clearly the best that proposition \ref{prop:sqrt3}
can achieve is to prove that a sequence
of quasimodes has an overlap of up to $1/\sqrt{3}$ with a subsequence of
true eigenfunctions. In order to prove that a sequence of quasimodes
converges {\em fully\/} towards a sequence of eigenvalues of $H_\Theta$ we 
need to make some assumptions on the spectrum of the Laplacian. 
Sufficient conditions for this and a precise statement of 
the result (theorem \ref{thm:seba}) are given in section \ref{sec:vier}.

The history of the study of the spectral properties of differential operators
perturbed by the addition of a delta scatterer goes back
at least to \cite{kro:qme}, in which a one-dimensional lattice of
delta interactions was used to model an electron moving in a
crystal lattice. A comprehensive historical review is given in the
appendix to the book \cite{alb:spd}.


Part of our interest in the subject comes from the \v{S}eba billiard 
which was introduced in \cite{seb:wcs}. In this work, a hard-walled
rectangular billiard with a potential supported at a single point was 
considered.
In terms of classical dynamics, the motion is integrable, since only
a zero-measure set of trajectories meet the point at which the 
potential is supported. However, diffraction effects are introduced
when one considers the quantum spectrum of the corresponding 
Schr\"odinger operator.

\v{S}eba billiards have become important since the observation 
\cite{seb:wci,shi:lsd} that the quantum spectral statistics 
belong to a new universality class, different from the classes from
 random matrix theory conjecturally associated to chaotic dynamical systems 
\cite{boh:ccq, cas:otc} or the statistics of a Poisson process conjecturally
associated
to fully-integrable dynamical systems \cite{ber:lcr}.  It is now known 
that general integrable systems perturbed by the addition of such a localised
scatterer also belong to the same universality class \cite{alb:wcq}, as
do quantum Neumann star graphs \cite{ber:tps,ber:sgs}. Characteristic
features of the spectral statistics of this universality class are an
exponential decay of large level spacings, together with level repulsion.

Several analytical studies of these spectral statistics have been made
\cite{alb:wcq,bog:sct, bog:ss, bog:nnd, rah:ssr, rah:ppc}. Typically,
a key feature of these arguments is the assumption of Poissonian behaviour
for the eigenvalues of the billiard table without scatterer, 
a conjectured consequence of the integrable dynamics
(the Berry-Tabor conjecture) \cite{ber:lcr}.

In the final section of this article we apply theorem \ref{thm:seba} to
the original \v{S}eba billiard. Our final result is a proof that there
exists a subsequence of eigenfunctions of the \v{S}eba billiard that
become localised on a pair of consecutive eigenfunctions of the
unperturbed billiard, if the spectrum of the unperturbed billiard
satisfies an assumption which is consistent with the Berry-Tabor
conjecture. 

This result is a 
rigorous derivation of a formal argument first proposed in \cite{ber:iws}
and mirrors a related result proved for quantum graphs with a star-shaped
connectivity \cite{ber:nqe}. These so-called quantum star graph can be 
considered as a singular perturbation of a disconnected set of one-dimensional 
bonds, each supporting a wave-function. In \cite{ber:nqe} the existence
of subsequences of eigenfunctions that become localised on a pair of bonds
was proved. This is exactly analagous to the localisation onto a pair
of unperturbed billiard eigenfunctions in theorem \ref{thm:seba}. In 
both \cite{ber:nqe} and theorem \ref{thm:seba} the main idea of the proof
is to show localisation in an eigenfunction with eigenvalue lying between
two closely-spaced eigenvalues of the unperturbed problem.

\section{Realisation of the perturbed operator} \label{sec:zwei}

Let $\curlyM$ be a compact 2-dimensional Riemannian manifold, possibly
with piecewise-smooth
boundary, and let $\Delta$ be a self-adjoint Laplacian on $\curlyM$.

The realisation of the operator formally defined by
\begin{equation} \label{eq:formal:H}
H = -\Delta + c\delta(x-p),
\end{equation}
where $p\in\curlyM$ and $\delta$ is the Dirac delta function, using the
theory of self-adjoint extensions is given in many places in the literature.
We refer the reader to \cite{cdv:plI, zor:psa} for the details. Here
we recapitulate only that which is necessary to fix notations.
We denote by $\|\cdot\|$ and $\langle\cdot,\cdot\rangle$ the norm
and inner product of $L^2(\curlyM)$.

Since $\curlyM$ is compact, $-\Delta$ has a complete basis of 
eigenfunctions, $\Phi_j$, with corresponding eigenvalues $E_j$
which we write in non-decreasing order.

We will remove from the list of eigenvalues any $E_j$ for which $\Phi_j(p)=0$.
Such eigenfunctions are not affected by a delta-scatterer at $p$, and
so it is convenient to exclude them from the spectrum. 
This further allows us to assume that the spectrum $\{E_j\}$ is simple, 
without losing generality. 

To see this, consider an eigenspace of dimension
$r>1$ spanned by the eigenfunctions $\{\tilde\phi_1,\ldots\tilde\phi_r\}$.
Then the vectors $(\tilde\phi_1(p),\ldots,\tilde\phi_r(p))^{\rm T}$ and
$(R,0,\ldots,0)^{\rm T}$ in $\C^r$, where 
\begin{equation}
  |R|^2=\sum_{i=1}^r |\tilde\phi_i(p)|^2,
\end{equation}
have identical norm. This means that we can find a unitary $r\times r$ 
matrix mapping the first vector to the second. Multiplying $U$ by the
vector of eigenfunctions 
$(\tilde\phi_1,\ldots,\tilde\phi_r)^{\rm T}$  leads to a new basis for the 
eigenspace, in which all but the first eigenfunction vanishes at the point
$p$, and the corresponding eigenvalue is counted with multiplicity one.

The resulting spectrum is therefore ordered so that
\begin{equation}
E_1<E_2 < E_3 \cdots
\end{equation}


We will frequently use Weyl's law with remainder estimate
\cite{ivr:tst}:
\begin{equation} \label{eq:weyl}
N(E)\coloneq \sum_{E_j\leq E} |\Phi_j(p)|^2 = \frac{E}{4\pi}
+\Ord(E^{1/2}),
\end{equation}
where the implied constant\footnote{The notations $f=\Ord(g)$ and $f\ll g$ 
both mean that there exists a positive constant $C$ (the ``implied constant'')
such that $f\leq Cg$.}
may depend on the position of the point $p\in\curlyM$.

Define
\begin{equation}
g_{\pm\rmi}(x)\coloneq \sum_{j=1}^\infty \frac{\Phi_j(x)\overline{\Phi_j(p)}}
{E_j\mp \rmi}.
\end{equation}
Then $g_{\pm\rmi}\in L^2(\curlyM)$ and in fact they are the Green
functions for the resolvent of $-\Delta$ at the imaginary energies $\pm\rmi$,
satisfying
\begin{equation}
\langle f,g_{\pm \rmi}\rangle = (-\Delta\pm\rmi)^{-1}f(p).
\end{equation}
In particular, 
\begin{align} \nonumber
      \langle \Phi_j, g_{\pm\rmi}\rangle 
&=((-\Delta\pm\rmi)^{-1}\Phi_j)(p)\\
&=\frac{\Phi_j(p)}{E_j\pm\rmi}, \label{eq:lem:eins}
\end{align}
which will be useful to know later.


Let 
\begin{equation}
 \curlyD_p \coloneq \left\{ f\in \dom(\Delta)
 : \langle f,\delta_p\rangle =0 \right\},
\end{equation}
and define the operator $H_0$ with domain $\curlyD_p$ by
\begin{equation}
H_0 : f \mapsto -\Delta f.
\end{equation}
$H_0$ is a symmetric, but not self-adjoint operator. In fact its
deficiency subspaces are spanned by $g_{\pm\rmi}$.

It follows from the {von Neumann} theory \cite{neu:aeh} that, 
\begin{equation}
\dom(H_0^\asterisk)=\curlyD_p \oplus \mathop{\rm span}\{ g_\rmi, 
g_{-\rmi}\}.
\end{equation}
Since the deficiency indices
are equal, $H_0$ possesses self-adjoint extensions, constructed as
follows.

First of all, note that we can write for $\psi\in\dom(H_0^*)$,
\begin{equation}
  \psi = \hat\psi + a_+(\psi) g_\rmi + a_-(\psi) g_{-\rmi},
\end{equation}
where $\hat\psi\in\curlyD_p$ and $a_\pm(\psi)\in\C$.
In fact we have
\begin{equation}
H_0^\asterisk\psi = 
H_0 \hat\psi + \rmi a_+(\psi) g_\rmi -\rmi  a_-(\psi) g_{-\rmi}. 
\end{equation}
Since the deficiency indices of $H_0$ are both equal to 1, 
there is a 1-parameter
family of self-adjoint extensions, $H_\Theta$, $0 < \Theta \leq 2\pi$ with
\begin{equation}
\dom(H_\Theta) = \{\psi\in\dom(H_0^*) : a_-(\psi) = -\rme^{\rmi \Theta}
a_+(\psi)\}.
\end{equation}
We take the self-adjoint operator $H_\Theta$ to be the realisation
of the formal operator \eqref{eq:formal:H}.

\section{Quasimodes} \label{sec:drei}

\subsection{Definitions and basic properties}

Let $H$ be a self-adjoint operator in a Hilbert
space, without continuous spectrum.

\begin{definition}
  A quasimode of $H$ with discrepancy $d$ is a pair 
$(\psi,\mu)\in\dom(H)\times\R$ such that
\begin{equation}
\| (H-\mu)\psi\| \leq d \| \psi \|.
\end{equation}
\end{definition}

We are interested in the situation when the quasi-eigenvalue $\mu$
and quasi-eigenfunction $\psi$ approximate true eigenvalues $\lambda_j$ and
eigenfunctions $\phi_j$ of $H$. In this direction, the following classical
results apply (see, e.g. \cite{laz:Kts,mar:ql})

For a quasimode with discrepancy $d$, the interval $[\mu-d,\mu+d]$
contains at least one eigenvalue of $H$.

If we consider instead, the interval $[\mu-M,\mu+M]$ where $M>0$, then
\begin{equation} \label{eq:quasi:proj}
  \sum_{\lambda_j\not\in[\mu-M,\mu+M]} |\langle \psi,\phi_j\rangle|^2 \leq
\frac{d^2}{M^2}\|\psi\|^2.
\end{equation}

In particular, if $\psi$ is normalised, and the interval $[\mu-M,\mu+M]$
contains only a single eigenvalue with eigenfunction $\phi$, 
then there is a phase $\chi\in[0,2\pi)$ such that
\begin{equation} \label{eq:quasi:approx}
 \| \phi-\rme^{\rmi\chi}\psi \| \leq \frac{2d}M.
\end{equation}

These results will be the main tools by which we relate the quasimodes 
constructed in the next subsection to the eigenfunctions and eigenvalues of
$H_\Theta$.

\subsection{Quasimodes of delta perturbations}
Let $I\subseteq\R$ be a finite interval containing at least one point
$E_j$ of the spectrum of $-\Delta$. Let $\sigma\in[0,1]$. We will
associate to the interval $I$ a family of quasimodes parametrised
by $\sigma$.

We first define
\begin{equation} \label{eq:better:qmodes}
\qqqef{z}\coloneq
 \sum_{E_j\in I} \frac{\overline{\Phi_j(p)}}{E_j-z}\Phi_j
+\frac \sigma{1-\rme^{\rmi\Theta}}\curlyP_{I^{\rm c}}\!\left( g_\rmi-\rme^{\rmi\Theta}
g_{-\rmi}\right), 
\end{equation}
where $\curlyP_S$ is the spectral projection operator onto the
set $S$,
\begin{equation}
\curlyP_Sf\coloneq\sum_{E_j\in S} \langle f,\Phi_j\rangle \Phi_j,
\end{equation}
and $I^{\rm c}$ is the complement to $I$. We have the following:
\begin{lemma} \label{lem:quasi:norm}
  For $z\neq E_j$ for any $E_j\in I$, the function $\qqqef{z}$ satisfies
\begin{equation}
\|\qqqef{z}\|^2 = \zeta_I(2,z) + \sigma^2\sum_{E_j\not\in I}\left( E_j-
\frac{\sin\Theta}{1-\cos\Theta}\right)^2\frac{|\Phi_j(p)|^2}{(1+E_j^2)^2},
\end{equation}
with the second term being bounded by a constant independent of $I, z$ 
and $\sigma\in[0,1]$.
\end{lemma}
\proof
We have
\begin{equation} \label{eq:nuovo:drei}
\|\qqqef{z}\|^2 = \sum_{E_j\in I}\frac{|\Phi_j(p)|^2}{(E_j-z)^2} +
\frac{\sigma^2}{|1-\rme^{\rmi\Theta}|^2}\left\| \curlyP_{I^{\rm c}}(g_\rmi -
\rme^{\rmi\Theta}g_{-\rmi})\right\|^2.
\end{equation}
By (\ref{eq:lem:eins}) we get
\begin{align}
  \langle \Phi_j, g_\rmi - \rme^{\rmi\Theta} g_{-\rmi}\rangle &= \nonumber
\left( \frac1{E_j+\rmi} - \frac{\rme^{-\rmi\Theta}}{E_j-\rmi}\right)\Phi_j(p)\\
&=\frac{E_j(1-\rme^{-\rmi\Theta})-\rmi(1+\rme^{-\rmi\Theta})}{1+E_j^2}
\Phi_j(p) \nonumber \\
&=(1-\rme^{-\rmi\Theta})\left( \frac{E_j}{1+E_j^2} - \frac{\sin\Theta}
{1-\cos\Theta}\frac1{1+E_j^2}\right)\Phi_j(p), \label{eq:useful}
\end{align}
using
\begin{equation} \label{eq:nice:identity}
  \rmi\frac{1+\rme^{\rmi\Theta}}{1-\rme^{\rmi\Theta}} = \frac{-\sin\Theta}
{1-\cos\Theta}.
\end{equation}
By Parseval's identity
\begin{equation}
 \frac1{|1-\rme^{\rmi\Theta}|^2 }   \label{eq:parseval}
\left\| \curlyP_{I^c}(g_\rmi - \rme^{\rmi\Theta}g_{-\rmi})\right\|^2
=\sum_{E_j\not\in I} \left( E_j - \frac{\sin\Theta}
{1-\cos\Theta}\right)^2\frac{|\Phi_j(p)|^2}{(1+E_j^2)^2}.
\end{equation}
Finally, to show that the right-hand side of \eqref{eq:parseval} is finite
and does not depend on $I$, we
observe that it is bounded by
\begin{equation}
  \sum_{j=1}^\infty \left( E_j - \frac{\sin\Theta}
{1-\cos\Theta}\right)^2\frac{|\Phi_j(p)|^2}{(1+E_j^2)^2} = 
\int_0^\infty \left(\frac1{1+t^2}
\left( t-\frac{\sin\Theta}{1-\cos\Theta}\right)\right)^2\,\rmd N(t),
\end{equation}
writing the sum as a Riemann-Stieltjes integral. The spectral counting function
$N(t)$ was defined in \eqref{eq:weyl}. Integrating by parts, we get
\begin{equation}
\label{eq:weyl:step}
  \sum_{E_j\not\in I} \left( E_j - \frac{\sin\Theta}
{1-\cos\Theta}\right)^2\frac{|\Phi_j(p)|^2}{(1+E_j^2)^2} \leq 
-\int_0^\infty \frac{\rmd}{\rmd t} \left(\frac1{1+t^2}\left(
 t-\frac{\sin\Theta}{1-\cos\Theta}\right)\right)^2 N(t) \,\rmd t.
\end{equation}
Since $N(t)\ll t$ by Weyl's law, we see that
the integral in \eqref{eq:weyl:step} is a finite 
constant. \finire

Let $\qqev=\qqev(\sigma,I)$ be a solution to
\begin{equation} \label{eq:qqev}
\sum_{E_j\in I} \frac{|\Phi_j(p)|^2}{E_j-\qqev} = \frac \sigma{1-\rme^{\rmi\Theta}}
\curlyP_I\!\left( g_\rmi - \rme^{\rmi \Theta}g_{-\rmi}\right)(p).
\end{equation}
Then the pair $(\qqef,\qqev)$ is a quasimode for $H_\Theta$, 
where $\qqef\coloneq\qqqef{{\qqev}}$. 
This follows from the following proposition:

\begin{proposition} \label{prop:quasi_discrep}
  The function $\qqef$ belongs to $\dom(H_\Theta)$ and satisfies
\begin{equation}
  \| (H_\Theta-\qqev)\qqef\|^2  = (1-\sigma)^2 \zeta_I(0,\mu)
 + \sigma^2\sum_{E_j\not\in I} \left( 1+E_j\qqev +\frac{\sin\Theta}{1-\cos\Theta}
(E_j-\qqev)\right)^2\frac{|\Phi_j(p)|^2}{(1+E_j^2)^2}.
\end{equation}
\end{proposition}
\proof
First of all, let us prove that $\qqef\in\dom(H_\Theta)$.

We can write
\begin{equation}
\qqef = \sum_{E_j\in I} \frac{\overline{\Phi_j(p)}}{E_j-\qqev}\Phi_j -
\frac \sigma{1-\rme^{\rmi\Theta}}\curlyP_I(g_\rmi-\rme^{\rmi\Theta}g_{-\rmi}) +
\frac \sigma{1-\rme^{\rmi\Theta}}(g_\rmi-\rme^{\rmi\Theta}g_{-\rmi}).
\end{equation}
Using \eqref{eq:useful} we can express this as
\begin{equation} \label{eq:quasi:hat_tilde}
\qqef = \hatqqef+\frac \sigma{1-\rme^{\rmi\Theta}}\left( g_\rmi - 
\rme^{\rmi \Theta}g_{-\rmi}\right),
\end{equation}
where 
\begin{equation}
\hatqqef(x)\coloneq \sum_{E_j\in I}\left( \frac1{E_j-\qqev}-
\frac{\sigma E_j}{1+E_j^2} + \frac{\sigma \sin\Theta}{1-\cos\Theta}
\frac1{1+E_j^2}\right)\overline{\Phi_j(p)}\Phi_j(x).
\end{equation}
Now observe that due to the definition \eqref{eq:qqev} of
$\qqev$, $\hatqqef(p)=0$, so $\hatqqef\in\curlyD_p$.
Thus \eqref{eq:quasi:hat_tilde} justifies the 
assertion $\qqef\in\dom(H_\Theta)$.

Since $H_0^*$ is an extension of $H_\Theta$, we have
\begin{equation} \label{eq:nuovo:zero}
  (H_\Theta - \qqev) \qqef = (H_0 - \qqev)\hatqqef + \frac \sigma{1-
\rme^{\rmi\Theta}}\left( (\rmi-\qqev)g_\rmi +\rme^{\rmi\Theta}
(\rmi+\qqev)g_{-\rmi}\right).
\end{equation}
Now,
\begin{align}
  (H_0-\qqev)\hatqqef &= \sum_{E_j\in I} \left( 1 - \frac{\sigma E_j(E_j-\qqev)}
{1+E_j^2} + \frac{\sigma\sin\Theta}{1-\cos\Theta} \frac{E_j-\qqev}{1+E_j^2}
\right)\overline{\Phi_j(p)}\Phi_j \nonumber \\
&= \sum_{E_j\in I} \left( 1 + (1-\sigma)E_j^2 +\sigma E_j\qqev
 + \frac{\sigma\sin\Theta}{1-\cos\Theta} (E_j-\qqev)
\right)\frac{\overline{\Phi_j(p)}}{{1+E_j^2}}\Phi_j. \label{eq:nuovo:eins}
\end{align}
Using (\ref{eq:lem:eins}) we find
\begin{align}
 \left \langle \Phi_j , (\rmi-\qqev)g_\rmi +\rme^{\rmi\Theta}
(\rmi+\qqev)g_{-\rmi} \right\rangle &=
\left( -\frac{\rmi+\qqev}{E_j+\rmi}\Phi_j(p)+\rme^{-\rmi\Theta}\frac{-\rmi
+\qqev}{E_j-\rmi}\Phi_j(p)\right) \\
&=\frac{1-\rme^{-\rmi\Theta}}{1+E_j^2}\left( -(1+E_j\qqev)-\frac{\sin\Theta}
{1-\cos\Theta}(E_j-\qqev)\right)\Phi_j(p), \nonumber
\end{align}
again using \eqref{eq:nice:identity}. This leads to
\begin{equation}
\frac{\sigma}{1-\rme^{\rmi\Theta}}(H_\Theta-\qqev)(g_\rmi - \rme^{\rmi\Theta} 
g_{-\rmi}) = -\sigma\sum_{j=1}^\infty 
\left( 1+E_j\qqev + \frac{\sin\Theta}{1-\cos\Theta}
(E_j-\qqev)\right) \frac{\overline{\Phi_j(p)}}{1+E_j^2}\Phi_j,
\end{equation}
and combining this with \eqref{eq:nuovo:eins}, we get
\begin{align}
  (H_\Theta-\qqev) \qqef 
              &= (1-\sigma)\sum_{E_j\in I} \overline{\Phi_j(p)} \Phi_j\nonumber \\ &
\qquad\qquad -\sigma\sum_{E_j\not\in I}\left( 1+E_j\qqev+\frac{\sin\Theta}
{1-\cos\Theta}(E_j-\qqev)\right) \frac{\overline{\Phi_j(p)}}{1+E_j^2}\Phi_j.
\label{eq:nuovo:zwei}
\end{align}
Since the summations in \eqref{eq:nuovo:zwei} are over disjoint
sets it is easy to calculate the norm:
\begin{multline}
  \| (H_\Theta-\qqev)\qqef\|^2  = (1-\sigma)^2 \sum_{E_j\in I} |\Phi_j(p)|^2
 \\+ \sigma^2\sum_{E_j\not\in I} \left( 1+E_j\qqev +
\frac{\sin\Theta}{1-\cos\Theta}
(E_j-\qqev)\right)^2\frac{|\Phi_j(p)|^2}{(1+E_j^2)^2}.
\end{multline}
\finire
\subsubsection{Proof of theorem \ref{thm:quasi_discrep}}
The first part of the theorem follows from lemma \ref{lem:quasi:norm} 
and proposition \ref{prop:quasi_discrep} and the definition of a quasimode.



For the second part, let $\qqev_1\neq \qqev_2$ be two solutions of 
\eqref{eq:qqev}. We have
\begin{align} \nonumber
    \left\langle \qqqef{{\qqev_1}}, \qqqef{{\qqev_2}} \right\rangle &=
\Big\langle \sum_{E_j\in I} \frac{\Phi_j(p)}{E_j-\qqev_1}\Phi_j,
\sum_{E_j\in I} \frac{\Phi_j(p)}{E_j-\qqev_2}\Phi_j\Big\rangle + \frac{\sigma^2}
{|1-\rme^{\rmi\Theta}|^2}\left\| \curlyP_{I^{\rm c}}(g_\rmi - \rme^{\rmi\Theta}
g_{-\rmi})\right\|^2 \\
&=\sum_{E_j\in I}\frac{|\Phi_j(p)|^2}{(E_j-\qqev_1)(E_j-\qqev_2)} +
\sigma^2 \sum_{E_j\not\in I} \left( E_j - \frac{\sin\Theta}{1-\cos\Theta}\right)^2
\frac{|\Phi_j(p)|^2}{(1+E_j^2)^2},
\end{align} 
using \eqref{eq:parseval}. By elementary algebra,
\begin{equation}
\frac1{(E_j-\qqev_1)(E_j-\qqev_2)} = \frac1{\qqev_1-\qqev_2}\left( 
\frac1{E_j-\qqev_1} - \frac1{E_j-\qqev_2}\right),
\end{equation}
so, since by \eqref{eq:qqev}
\begin{equation}
\sum_{E_j\in I} \frac{|\Phi_j(p)|^2}{E_j-\qqev_1} = 
\sum_{E_j\in I} \frac{|\Phi_j(p)|^2}{E_j-\qqev_2},
\end{equation}
we get
\begin{equation}
 \sum_{E_j\in I}\frac{|\Phi_j(p)|^2}{(E_j-\qqev_1)(E_j-\qqev_2)} = 0.
\end{equation}
\finire

\subsubsection{Controlling the discrepancy of quasimodes}

By tuning the parameter $\sigma$, and choosing the interval $I$ accordingly,
we can find fix quasimodes with particular properties. In the previous
section we have seen that sets of quasimodes with $\sigma=0$ are orthogonal.
We are particularly interested in when the discrepancy is small.
In this subsubsection we prove corollary \ref{cor:arbitrary:small}
that quasimodes with $\sigma=1$ can be made arbitrarily precise, and
that quasimodes with $\sigma=0$ also can have
a simple bound for the discrepancy.

\begin{proof}[Proof of corollary \ref{cor:arbitrary:small}]
Choosing $I=[0,T]$ for $T>E_1$ with $\sigma=1$ gives, by theorem 
\ref{thm:quasi_discrep}, that the discrepancy of $\psi_{1,I}$ satisfies
\begin{equation}\label{eq:arb:disc} d^2 \|\psi_{1,I}\|^2 = 
\sum_{E_j\geq T} 
\left( 1+E_j\qqev +\frac{\sin\Theta}{1-\cos\Theta}
(E_j-\qqev)\right)^2\frac{|\Phi_j(p)|^2}{(1+E_j^2)^2}.\end{equation}
By lemma \ref{lem:quasi:norm} we see that the norm of $\psi_{1,I}$ is 
bounded away from $0$ by a constant, so that the asymptotics for $d$ are given
by the term on the right-hand side of \eqref{eq:arb:disc}. 
Using Weyl's law, we can estimate
\begin{equation}
  \sum_{E_j\geq T} 
\left( 1+E_j\qqev +\frac{\sin\Theta}{1-\cos\Theta}
(E_j-\qqev)\right)^2\frac{|\Phi_j(p)|^2}{(1+E_j^2)^2} \ll \frac{\qqev^2}{T},
\end{equation}
which can be made arbitrarily small by increasing $T$.

For the second part with $\sigma=0$, we have
\begin{equation}
  d^2\| \psi_{0,I}\|^2 = \sum_{E_j\in I} |\Phi_j(p)|^2.
\end{equation}
We observe that splitting the sum in \eqref{eq:qqev} leads to
\begin{equation}
  \label{eq:split:evsum}
  \sum_{\substack{E_j\in I\\E_j>\mu}}\frac{|\Phi_j(p)|^2}{E_j-\mu} =
  \sum_{\substack{E_j\in I\\E_j<\mu}}\frac{|\Phi_j(p)|^2}{\mu-E_j}.
\end{equation}
Denote by $E_+$ and $E_-$ the largest and smallest points of the spectrum
$(E_j)_{j=1}^\infty$ lying in the interval $I$. Then
\begin{equation}
  \sum_{\substack{E_j\in I\\E_j>\mu}} |\Phi_j(p)|^2 \leq
(E_+-\mu)\sum_{\substack{E_j\in I\\E_j>\mu}}\frac{|\Phi_j(p)|^2}{E_j-\mu} 
\end{equation}
and
\begin{equation}
  \sum_{\substack{E_j\in I\\E_j<\mu}} |\Phi_j(p)|^2 \leq
(\mu-E_-)\sum_{\substack{E_j\in I\\E_j<\mu}}\frac{|\Phi_j(p)|^2}{\mu-E_j}. 
\end{equation}
Adding these inequalities, and using \eqref{eq:split:evsum}, we get
\begin{equation}
  \sum_{E_j\in I} |\Phi_j(p)|^2 \leq (E_+-E_-) 
\sum_{\substack{E_j\in I\\E_j>\mu}}\frac{|\Phi_j(p)|^2}{E_j-\mu} =
(E_+-E_-) \sum_{\substack{E_j\in I\\E_j<\mu}}\frac{|\Phi_j(p)|^2}{\mu-E_j}.
\end{equation}
Since $E_+-E_-\leq \ell(I)$ we get
\begin{equation}
  2d^2\|\psi_{0,I}\|^2 \leq \ell(I) \sum_{E_j\in I} \frac{|\Phi_j(p)|^2}
{|E_j-\mu|}.
\end{equation}
Finally,
\begin{equation}
  \ell(I) \sum_{E_j\in I} \frac{|\Phi_j(p)|^2}{|E_j-\mu|} \leq
\ell(I)^2 \sum_{E_j\in I} \frac{|\Phi_j(p)|^2}{(E_j-\mu)^2}
=\ell(I)^2\|\psi_{0,I}\|^2,
\end{equation}
noting that $|E_j-\mu|\leq \ell(I)$ for $\mu\in I$.

We now consider the case with $\sigma=0$, and $I$ containing only the
two levels $E_j$, $E_{j+1}$. We can solve \eqref{eq:qqev} directly to get
\begin{equation}
  \mu = \frac{|\Phi_{j+1}(p)|^2 E_j + |\Phi_j(p)|^2 E_{j+1}}
{|\Phi_{j+1}(p)|^2 + |\Phi_j(p)|^2}. 
\end{equation}
Substituting this value of $\mu$ into the definition of $\psi_{0,I}$ we 
get
\begin{align}
  \psi_{0,I} &= \frac{\overline{\Phi_j(p)}}{E_j-\mu}\Phi_j + 
\frac{\overline{\Phi_{j+1}(p)}}{E_{j+1}-\mu}\Phi_{j+1} \\
&=\frac{|\Phi_{j+1}(p)|^2 + |\Phi_j(p)|^2}{E_{j+1}-E_j}\left( 
\frac{-1}{\Phi_j(p)}\Phi_j + \frac1{\Phi_{j+1}(p)}\Phi_{j+1}\right).
\end{align}
So,
\begin{align}
  \|\psi_{0,I}\|^2 &= d^2\|\psi_{0,I}\|^2\frac{|\Phi_{j+1}(p)|^2 + 
|\Phi_j(p)|^2}{(E_{j+1}-E_j)^2}\left( 
\frac{1}{|\Phi_j(p)|^2}+ \frac1{|\Phi_{j+1}(p)|^2}\right) \\
&\geq \frac{4 d^2 \|\psi_{0,I}\|^2}{\ell(I)^2},
\end{align}
using the fact that
\begin{equation}
  {|\Phi_{j+1}(p)|^2 + |\Phi_j(p)|^2}\left( 
\frac{1}{|\Phi_j(p)|^2}+ \frac1{|\Phi_{j+1}(p)|^2}\right)
=2 + \left|\frac{\Phi_{j+1}(p)}{\Phi_j(p)}\right|^2 + 
 \left|\frac{\Phi_{j}(p)}{\Phi_{j+1}(p)}\right|^2 \geq 4.
\end{equation}
\end{proof}

  The existence of arbitrarily precise quasimodes
 can be used to give a new proof of the often-used 
representation for eigenvalues and eigenfunctions of rank-one perturbations
(see e.g. \cite{alb:wcq, cdv:plI, rah:ssr, shi:wcq})

\begin{theorem} \label{thm:true:modes}
The solutions $\lambda$ to the equation
\begin{equation} \label{eq:seba:eigenvalue}
  \sum_{j=1}^\infty \left( \frac1{E_j-\lambda} - \frac{E_j}{1+E_j^2}\right)
|\Phi_j(p)|^2=\frac{\sin\Theta}{1-\cos\Theta}\sum_{j=1}^\infty 
\frac{|\Phi_j(p)|^2}{1+E_j^2},
\end{equation}
are eigenvalues of $H_\Theta$ with corresponding eigenfunctions given by
\begin{equation} \label{eq:seba:eigenfunction}
  \phi(x)=\sum_{j=1}^\infty \frac{\overline{\Phi_j(p)}}{E_j-\lambda}\Phi_j(x).
\end{equation}
\end{theorem}
\noindent
Note that the left-hand side of \eqref{eq:seba:eigenvalue} converges pointwise,
and \eqref{eq:seba:eigenfunction} converges in $L^2(\Omega)$.

By analysing the resolvent, it is possible to extend theorem 
\ref{thm:true:modes} to get the following \cite[Theoreme 2]{cdv:plI},
\begin{theorem} \label{thm:no:other}
  Apart from the solutions to \eqref{eq:seba:eigenvalue}, there are no
other points of the spectrum of $H_\Theta$ in any of
the intervals $(E_M,E_{M+1})$.
\end{theorem}

\section{Localisation results} \label{sec:vier}

In this section we will consider the extent to which eigenfunctions
of $H_\Theta$ can be approximated by quasimodes. In particular we will
focus on the quasimodes with $\sigma=0$. First we shall prove proposition
\ref{prop:sqrt3}, which is straightforward. Then we shall show that
strengthening the assumptions made on the spectrum of $\Delta$ leads to
a proof of full convergence.


\begin{proof}[Proof of proposition \ref{prop:sqrt3}]
The length of the interval $I$ is $\ell(I)=\Ec-\Eb$. 
Let $M=\min\{\Ed-\mu,\mu-\Ea\}\geq \min\{\Ed-\Ec,\Eb-\Ea\}$.
By applying (\ref{eq:quasi:proj}) with this $M$ we get, 
\begin{equation} \label{eq:norm:squared}
  \sum_{\lambda_j\in[\Ea,\Ed]} |\langle\psi_{0,I},\phi_j\rangle|^2 \geq
\|\psi_{0,I}\|^2\left( 1-\frac{\ell(I)^2}{4\min\{\Ed-\Ec,\Eb-\Ea\}^2}\right).
\end{equation}
From theorem \ref{thm:no:other} there are only three eigenvalues of
$H_\Theta$ in the interval $[\Ea,\Ed]$.
It therefore follows that for at least one of these three eigenfunctions its
inner-product squared  with $\psi_{0,I}$ is at least
$\frac13$ of the right-hand side of \eqref{eq:norm:squared}.
\end{proof}

We now consider how to improve proposition \ref{prop:sqrt3} at the
expense of making further assumptions about the spectrum of $-\Delta$.
For simplicity we will focus henceforth on the choice of parameter 
$\Theta=\pi$.

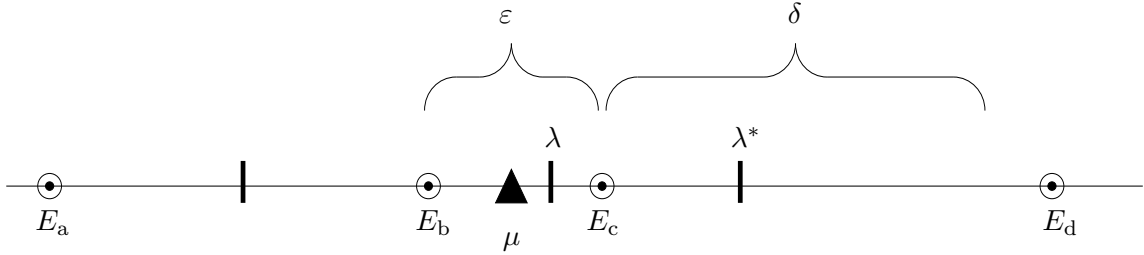
\begin{figure}[htbp]
\begin{center}
\setlength{\unitlength}{5cm}
\begin{picture}(3,0.8)
\put(0,0.1){\includegraphics[angle=0,width=15.0cm,height=2.2cm]{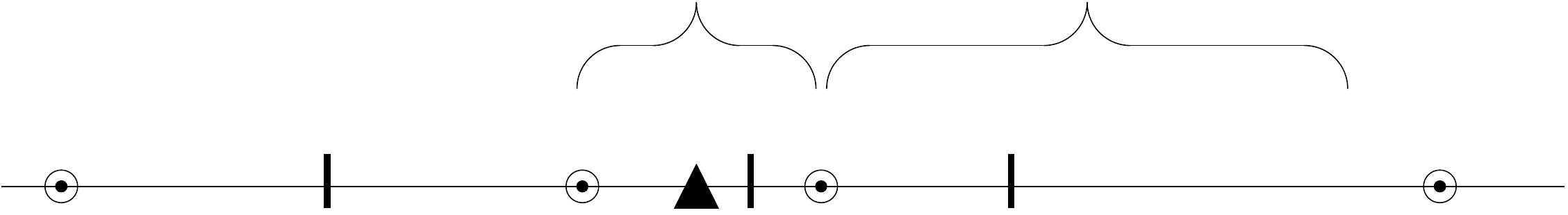}}
\put(1.08,0.04){$\Eb$}
\put(1.3,0.59){$\epsilon$}
\put(2.73,0.04){$\Ed$}
\put(1.53,0.04){$\Ec$}
\put(2.06,0.59){$\delta$}
\put(1.31,0.00){$\qqev$}
\put(0.08,0.04){$\Ea$}
\put(1.42,0.26){$\lambda$}
\put(1.91,0.26){$\lambda^*$}
\end{picture}
\caption{Part of the spectrum of $-\Delta$ and $H_\pi$. Vertical bars
denote eigenvalues of $H_\pi$, and circles denote eigenvalues of $-\Delta$. 
The triangle $\qqev$ is a quasi-eigenvalue approximating
$\lambda$. See main text for further explanation.}
\label{fig:zero}
\end{center}
\end{figure}

In figure \ref{fig:zero} a cartoon of part of the spectrum of $H_\pi$ and
$-\Delta$ is displayed. Highlighted are four consecutive eigenvalues of
$-\Delta$, labelled $\Ea, \Eb, \Ec$ and $\Ed$, chosen so that $\Ec-\Eb
\leq \epsilon$. (The positions of all points depend on $\epsilon$.)

Between $\Eb$ and $\Ec$ is an eigenvalue, $\lambda$, of $H_\pi$. 

We find a quasimode $\psi_{0,I}$ associated to the interval 
$I=[\Eb,\Ec]$
with quasi-eigenvalue $\qqev$ 
approximating $\lambda$. By corollary \ref{cor:arbitrary:small}
the discrepancy of this
quasimode is no greater than $\epsilon/2$. Between $\Ec$ and $\Ed$ is 
another eigenvalue $\lambda^*$ of $H_\pi$. 
In order to be able to apply \eqref{eq:quasi:approx},
we need to be sure that $\lambda^*$ is not too close to $\Ec$. 
An argument to show that this is the case is given below.

The eigenvalue between $\Ea$ and $\Eb$ can be handled with a similar
method.

We shall make the following assumption on the spectral sequence of $\Delta$:

\begin{assumption} For some $0<q<1/2$ and $1<\rho < 2(1-q)$,
  there exists a sequence $(\epsilon_n)_{n=1}^\infty$, $\epsilon_n\downarrow
0$ such that for each $n$ there are four consecutive eigenvalues,
$\Ea(n)<\Eb(n)<\Ec(n)<\Ed(n)\ll
\epsilon_n^{-\rho}$, satisfying   
\begin{align} \nonumber
  \Ec-\Eb &\ll \epsilon_n \\
  \Ed-\Ec &\gg \epsilon_n^q \\
  \Eb-\Ea &\gg \epsilon_n^q. \nonumber
\end{align} as $n\to\infty$.
\label{ass:main}
\end{assumption}
 
Assumption \ref{ass:main} asserts that the positions of eigenvalues
of $-\Delta$ occur with the spacings as described above, and furthermore
that this does not happen too high up in the spectrum. This upper bound
is necessary
as a consequence of the non-uniform convergence in $\lambda$ of the series
in \eqref{eq:seba:eigenvalue}. In appendix B we show that assumption
\ref{ass:main} is satisfied almost surely if the sequence $(E_j)$ comes
from a Poisson process. In this sense, assumption \ref{ass:main} is
consistent with the Berry-Tabor conjecture, if $-\Delta$ is the Hamiltonian
corresponding to an integrable dynamical system.

We shall also assume a lower bound for the absolute values of the
eigenfunctions $\Phi_j$ at the point $p$.
\begin{assumption} \label{ass:extra}
  There exists a constant $c_0>0$ independent of $j$ such that 
\begin{equation}
|\Phi_j(p)|\geq c_0.
\end{equation}
\end{assumption}
\begin{remark}
  In fact we require only that assumption \ref{ass:extra} holds for
(possibly a subsequence of) the sequence of pairs $\Phi_{\rm b}(p)$
and $\Phi_{\rm c}(p)$ for eigenfunctions associated with the sequences of energy
levels $\Eb$ and $\Ec$ defined in assumption \ref{ass:main}.
\end{remark}
We recall that the spectral sequence is defined in such a way that
$\Phi_j(p)\neq 0$ for all $j$. Thus assumption \ref{ass:extra} disqualifies
subsequences of eigenfunctions converging to $0$ at the point $p$. 

This assumption reflects the fact that if $|\Phi_j(p)|$ becomes small,
two eigenvalues of $H_\pi$ will approach $E_j$. Then we would only be
able to prove that the quasi-eigenfunction approximates a certain
linear combination of these eigenfunctions of $-\Delta$, rather
than an actual eigenfunction. Assumption \ref{ass:extra} can be
relaxed slightly---see remark \ref{rem:extra} below.

\begin{theorem} \label{thm:seba}
Assume that the spectrum of $-\Delta$ satisfies assumptions \ref{ass:main} and
\ref{ass:extra}. Then
the sequence of quasimodes $\psi_{0,I}$ associated
to the sequence of intervals $I=[\Eb,\Ec]$ and $\mu\in I$, with
$\Ea,\ldots,\Ed$ as described in assumption \ref{ass:main},
after normalisation, converge in $L^2$ to a subsequence of true eigenfunctions
of $H_\pi$.
\end{theorem}
Let us fix a point $n$ of the sequence $(\epsilon_n)$ with 
$\epsilon_n=\epsilon$, and $I$ fixed as described in the statement of theorem
\ref{thm:seba}.

\begin{proof}[Proof of theorem \ref{thm:seba}] 
In order to use \eqref{eq:quasi:approx} we will employ partial
summation, to estimate the position of eigenvalues of $H_\pi$.
If $g$ is a smooth function, then
\begin{equation}
  \label{eq:sum:parts}
  \sum_{X\leq E_j \leq Y} g(E_j)|\Phi_j(p)|^2 = g(Y)N(Y) - g(X)N(X) - \int_X^Y
g'(t)N(t)\,\rmd t,
\end{equation}
where $N(t)$ has been defined in \eqref{eq:weyl}. Equation \eqref{eq:sum:parts}
may be proved by Riemann-Stieltjes integration. Let $\lambda^\asterisk$ 
be the solution
of \eqref{eq:seba:eigenvalue} lying between $\Ec$ and $\Ed$. Let
\begin{equation} \label{eq:g:def}
  g(t)\coloneq \frac1{t-\lambda^\asterisk} - \frac{t}{1+t^2}=
\frac{1+t\lambda^\asterisk}{(t-\lambda^\asterisk)(1+t^2)},
\end{equation}
and observe that $g(t)>0$ if $t>\lambda^\ast$ and $g(t)<0$ if 
$t<\lambda^\asterisk$. By \eqref{eq:seba:eigenvalue} we have
\begin{equation}
  \label{eq:lambdastar}
  0=\sum_{j=1}^\infty g(E_j)|\Phi_j(p)|^2 \leq g(\Ec)|\Phi_{\rm c}(p)|^2
+\sum_{E_j\geq \Ed}g(E_j)|\Phi_j(p)|^2.
\end{equation}
Now, by \eqref{eq:sum:parts},
\begin{align}
  \sum_{E_j\geq \Ed} g(E_j)|\Phi_j(p)|^2 &= -g(\Ed)N(\Ed)-\int_{\Ed}^\infty
g'(t)N(t)\,\rmd t \\
&=\frac1{4\pi}\int_{\Ed}^\infty g(t)\,\rmd t 
+\Ord\!\left( g(\Ed)\Ed^{1/2} + \int_{\Ed}^\infty |g'(t)|t^{1/2} \,
\rmd t\right), \nonumber
\end{align}
using \eqref{eq:weyl}.

Since
\begin{equation}
  g'(t)=\frac{-1}{(t-\lambda^*)^2} - \frac1{1+t^2} + \frac{2t^2}{(1+t^2)^2},
\end{equation}
we get 
\begin{align}
  \int_{\Ed}^\infty |g'(t)|t^{1/2} \,\rmd t &\sim
 \int_{\Ed}^\infty \frac{t^{1/2}}{(t-\lambda^*)^2}\,\rmd t \nonumber \\
&\leq \left( \frac{\Ed}{\Ed-\lambda^*} \right)^{1/2}
\int_{\Ed}^\infty \frac1{(t-\lambda^*)^{{3/2}}}\,\rmd t \nonumber \\
&\ll  \left( \frac{\Ed}{\Ed-\lambda^*} \right)^{1/2} \frac1{(\Ed -
\lambda^*)^{{1/2}}} \nonumber \\
&=\frac{\Ed^{1/2}}{\Ed-\lambda^*}.
\end{align}
We can also calculate
\begin{equation}
  \int_{\Ed}^\infty g(t)\,\rmd t =   \int_{\Ed}^\infty
\frac1{t-\lambda^*} - \frac{t}{1+t^2} \,\rmd t =
 -\ln\!\left(\frac{\Ed-\lambda^*}{\sqrt{1+\Ed^2}}\right).
\end{equation}
So we have
\begin{equation}
  \sum_{E_j\geq\Ed} g(E_j)|\Phi_j(p)|^2 = \frac{-1}{4\pi}
\ln\!\left(\frac{\Ed-\lambda^*}
{\sqrt{1+\Ed^2}}\right) + \Ord\!\left( \frac{\Ed^{1/2}}{\Ed-\lambda^*}
\right),
\end{equation}
in which the dominant term on the RHS is actually the error term.
We have, from \eqref{eq:lambdastar}
\begin{align}
  g(\Ec)|\Phi_{\rm c}(p)|^2 &\ll \frac{\Ed^{1/2}}{\Ed-\lambda^*} \\
\Rightarrow\qquad \frac{|\Phi_{\rm c}(p)|^2}{\Ec-\lambda^\asterisk} &\ll
 \frac{\epsilon^{-\rho/2}}{\Ed-\Ec-(\lambda^*-\Ec)} \nonumber \\
&\leq \frac{\epsilon^{-\rho/2}}{\epsilon^q - (\lambda^*-\Ec)},
\label{eq:seba:quattro}
\end{align}
implying the lower bound 
\begin{equation}
  \lambda^\asterisk - \Ec \gg \epsilon^{\rho/2+q}.
\end{equation}
To see this, observe that if $\lambda^\asterisk = \littleo(\epsilon^{q+\rho/2})$
then we would have from \eqref{eq:seba:quattro}
\begin{equation}
\frac{|\Phi_{\rm c}(p)|^2}{\Ec-\lambda^\asterisk} \ll \epsilon^{-\rho/2-q},
\end{equation}
a contradiction.

By the same method, we can establish the same bound for the solution
to \eqref{eq:seba:eigenvalue} between $\Ea$ and $\Eb$, and by theorem
\ref{thm:no:other} we deduce that there is an interval of size 
$M\asymp \epsilon^{q+\rho/2}$ about $\qqev$ such that $[\qqev-M,\qqev+M]$
contains only one eigenvalue of $H_\pi$. Since $q+\rho/2<1$, and
since the discrepancy of
$\psi_{0,I}$ is $\Ord(\epsilon)$, 
equation \eqref{eq:quasi:approx} allows us to conclude that 
the normalised quasimode differs from the true eigenfunction
associated to $\qqev$ by an amount which converges to $0$ as $\epsilon\to 0$.
\end{proof}
\begin{remark} \label{rem:extra}
  From the proof of theorem \ref{thm:seba} we see that we can relax
assumption \ref{ass:extra} to demanding only that $|\Phi_j(p)|\gg 
\epsilon^{r/2}$ with $0<r<1-q-\rho/2$.
However, in a generic situation this is unlikely to be achieved. In
appendix \ref{app:A} we show that for a badly-approximable position of
the point $p$ in a rectangle, the best possible bound is 
\begin{equation}
|\Phi_j(p)|\gg\frac1{E_j},
\end{equation}
which is not sufficiently slow.
\end{remark}

\section{Application to rectangular \v{S}eba billiards} \label{sec:fuenf}
In this section we will apply theorem \ref{thm:seba} to the original
\v{S}eba billiard \cite{seb:wcs}.
We consider a rectangular billiard $\Omega=(0,2a)\times(0,2b)\subseteq\R^2$ and
point $p=(a,b)$ at the centre of the billiard. However, we remark that
we could position $p$ at any point with co-ordinates that are rational
multiples of the side lengths without significant changes to the forthcoming
analysis. 

The eigenvalues of 
$-\Delta$, the Laplacian with Dirichlet boundary conditions, are
given by
\begin{equation}
E_{n,m}=\frac{\pi^2}4\left(\frac{n^2}{a^2}+\frac{m^2}{b^2}\right),
\end{equation}
where $n,m\in\N$, and the corresponding eigenfunctions are
\begin{equation}
\Phi_{n,m}(x,y) = \frac1{\sqrt{ab}}\sin\!\left(\frac{n\pi x}{2a}\right)
\sin\!\left(\frac{m\pi y}{2b}\right).
\end{equation}
If either $n$ or $m$ are even, then the symmetry of the problem forces
$\Phi_{n,m}(p)=0$. So for these values of $n$ and $m$, 
$\Phi_{n,m}\in\curlyD_p$,
and are automatically eigenfunctions of the extended operator $H_\pi$. 
We exclude these eigenvalues from the spectrum, as discussed in section
\ref{sec:zwei}.

Instead, we concentrate on the more interesting subsequence where
$n$ and $m$ are both odd, e.g. $n=2s+1$ and $m=2t+1$ with $s,t=0,1,2,\ldots$
Then we have 
\begin{equation}
\Phi_{s,t}(p)=\frac1{\sqrt{ab}}(-1)^{s+t},
\end{equation}
so that along this sequence assumption \ref{ass:extra} is satisfied.
The corresponding set of eigenvalues is given by
\begin{equation} \label{eq:ex:eigenvals}
  E_{s,t}=\pi^2\left(\frac{(s+\frac12)^2}{a^2}+
\frac{(t+\frac12)^2}{b^2}\right),\qquad\mbox{$s,t=0,1,2,\ldots$}
\end{equation}
For generic choices of $a$ and $b$ it is conjectured that the set of
values given by \eqref{eq:ex:eigenvals} behave statistically like the
event times of a Poisson process \cite{ber:lcr,sar:vai,esk:qfs,mar:sff}. 
Under assumption
\ref{ass:main} for the set of values \eqref{eq:ex:eigenvals}, theorem
\ref{thm:seba} asserts the existence of a subsequence $(j_n)\subseteq\N$
such that
\begin{equation}
\| \phi_{j_n} - \psi_n\| \to 0\qquad\mbox{as $n\to\infty$},
\end{equation}
where $\phi_{j}$ are eigenfunctions of $H_\pi$, and $\psi_n$ are
of the form
\begin{equation}
\psi_n = \frac1{\sqrt{2}}\left( \Phi_{j_{n}} +
(-1)^{\beta_n}\Phi_{j_n+1}\right),
\end{equation}
where $\beta_n$ can be $0$ or $1$ and depends on the relative
signs of $\Phi_{j_n}(p)$ and $\Phi_{j_n+1}(p)$.
So the subsequence $(\phi_{j_n})$ converges to a superposition of
two consecutive unperturbed eigenfunctions of $-\Delta$.

The consequences for this subsequence are most striking when one
considers the momentum representation. This is
is given by the Fourier transform;
\begin{equation}
  \label{eq:venticinque}
  \hat{\phi}_j(p_x,p_y) = \frac1{2\pi}\int_{-\infty}^\infty \!
\int_{-\infty}^{\infty} \rme^{-\rmi x p_x - \rmi y p_y} \phi_j (x,y)\,
\rmd x\rmd y.
\end{equation}
For an ergodic system, the quantum ergodicity theorem of \v{S}nirel'man,
Zelditch and Colin de Verdi\`{e}re \cite{sch:epe, zel:ude, cdv:eef}
would imply that the momentum representation of almost all eigenfunctions
equi-distributes around the circle of radius $\sqrt{E_j}$ 
as $j\to\infty$;
\begin{equation}
  \label{eq:ventisei}
  |\hat{\phi}_j(x)|^2 \to \frac1\pi\delta(|x|^2 - E_j)
\qquad\mbox{as $j\to\infty$,}
\end{equation}
where convergence in \eqref{eq:ventisei} is in the weak sense.
\v{S}eba billiards are not ergodic, but we see a very different behaviour
to \eqref{eq:ventisei} for the subsequence $(\phi_{j_n})$.

From Parseval's theorem, it follows that
\begin{equation}
  \label{eq:ventisette}
  \hat{\phi}_{j_n} - \hat{\psi}_n \to 0 \qquad\mbox{in $L^2$ norm.}
\end{equation}\
The momentum representation of the unperturbed eigenfunctions 
$\Phi_{n,m}$ is
\begin{equation}
\hat{\Phi}_{n,m}(p_x,p_y) = \frac{2\pi n m \sqrt{ab}}{(4p_x^2 a^2 - n^2\pi^2)
(4p_y^2b^2-m^2\pi^2)}\left((-1)^n\rme^{-2\rmi a p_x}-1\right)\left(
(-1)^m\rme^{-2\rmi b p_y}-1\right).
\end{equation}
%
Since 
\begin{equation}
  \frac{n\pi}{4p_x^2 a^2 - n^2\pi^2} = \frac12\left(
\frac1{2 p_x a - n\pi} - \frac1{2 p_x a + n \pi} \right),
\end{equation}
we re-scale and write
\begin{equation}
  \label{eq:ventinove}
    \hat{\Phi}_{n,m}(np_x,mp_y)= \frac{\pi\sqrt{ab}}{2nm}
\left(\delta_n(2 p_x a +\pi) - \delta_n(2p_xa -\pi)\right) \left(\delta_m(2p_y b
-\pi) -\delta_m(2 p_y b +\pi)\right),
 \end{equation}
where $\delta_n$ is the smoothed-delta function
\begin{equation}
  \label{eq:trenta}
  \delta_n(t)\coloneq\frac{1-\rme^{-\rmi n t}}{\pi \rmi t}.
\end{equation}
The function $\delta_n(t)$ converges weakly to $\delta(t)$ as $n\to\infty$.
Furthermore, it satisfies
\begin{equation}
  \label{eq:trentuno}
  |\delta_n(t)|^2 \sim \frac{2n}\pi \delta(t) \qquad\mbox{as $n\to\infty$.}
\end{equation}
Hence
\begin{equation}
  \label{eq:trentadue}
  nm| \hat{\Phi}_{n,m}(np_x,mp_y)|^2 \sim
{ab}\left(\delta(2 p_x a - \pi) + \delta(2 p_x a + \pi)\right)
\left(\delta(2 p_y b 
- \pi) +\delta(2 p_y b +\pi) \right)
\end{equation}
as $n,m\to\infty$. The momentum eigenfunction localises around the
4 points 
\begin{equation}
  \label{eq:trentaquattro}
  (p_x,p_y) = \left( \pm \frac{n\pi}{2a}, \pm \frac{m\pi}{2b}  \right),
\end{equation}
which satisfy $p_x^2+p_y^2 = E_{n,m}$.
Since $\psi_n$ is a 
superposition of $\Phi_{j_n}$ and $\Phi_{j_n+1}$,
the states in the subsequence $\hat\phi_{j_n}$ become localised around 8 
points,
which all lie on the circle with radius $\sqrt{E_j}$, very much in contrast to
the expected equi-distribution \eqref{eq:ventisei} for ergodic systems.
Numerical simulations illustrating this behaviour have been presented
in \cite{ber:iws}.
This localisation is, in some sense, analogous to the scarring phenomenon
which occurs in some chaotic systems. Since these states are not associated 
with
an unstable periodic orbit, they do not fall into the very precise definition
of a scar given in \cite{kap:sqc}. Rather they are localising around ghosts of 
departed tori of the unperturbed integrable system. Nevertheless
they cannot be explained simply by using torus quantisation, and so they 
provide a further example of the already rich behaviours in systems with
intermediate statistics.

\subsection*{Acknowledgements}
We are grateful to Gregory Berkolaiko, Jens Bolte, Yves Colin de Verdi\`ere
and Tom Spencer for interesting conversations about this work. We thank an
anonymous referee for suggesting an improvement to an earlier version
of corollary \ref{cor:arbitrary:small}.

This work has been supported by an EPSRC Senior Research Fellowship (JPK),
a Royal Society Wolfson Research Merit Award (JM) and
the National Sciences Foundation under research grant DMS-0604859 (BW).

The writing-up of the manuscript was completed during a visit of two
of the authors (JM \& BW) to the Max-Planck-Institut f\"ur Mathematik,
Bonn.



\appendix
\section{Non-constant unperturbed eigenfunctions at the position of
the scatterer} \label{app:A}
In order to consider what can happen when the 
value of the unperturbed eigenfunctions at the position of the scatterer
can vary, let us consider the rectangular billiard $\Omega$, with sides of
length $a$ and $b$, and Dirichlet boundary conditions.

The energy levels are given by
\begin{equation}
E=E_{n,m} = \pi^2\left(\frac{n^2}{a^2}+\frac{m^2}{b^2}\right),
\end{equation}
for $n,m\geq 1$ integers. 
\begin{lemma}
\begin{equation}
  \frac1{n^2m^2}\geq \frac{4\pi^4}{a^2 b^2}\frac1{E^2}.
\end{equation}
\end{lemma}
\dimostrazione
We have 
\begin{align} \nonumber
  0\leq \pi^4 \left( \frac{n^2}{a^2}-\frac{m^2}{b^2}\right)^2 &=
\frac{n^4\pi^4}{a^4}-2\frac{n^2m^2}{a^2 b^2}\pi^4+ \frac{m^4 \pi^4}{b^4}\\
&=E^2-4\frac{n^2m^2}{a^2 b^2}\pi^4,
\end{align}
and then re-arrange to get the required estimate.\finire

The eigenfunctions themselves are proportional to
\begin{equation}
 \sin\left(\frac{n\pi x}a\right) \sin\left(\frac{m\pi y}b\right).
\end{equation}
Let us choose the point $p=(x_p,y_p)\in\Omega$ so that $x_p/a$ and $y_p/b$ are
badly-approximable, in the sense that
\begin{equation}
\left| n\frac{x_p}a - r\right| \geq\frac{C}n\qquad\forall n,r\in \Z,
\end{equation}
(this is the best we can hope to do if we want to bound the eigenfunctions
away from $0$). Then
\begin{equation}
n\frac{x_p}a = r + \vartheta(n)
\end{equation}
where $\vartheta$ can depend on $x_p$ and $r$ and satisfies 
\begin{equation}
|\vartheta(n)|\gg \frac1n
\end{equation}
uniformly. Furthermore this bound is achieved if $r/n$ is a continued
fraction approximant to $x_p/a$. We get
\begin{equation}
\sin^2\left(\frac{n\pi x_p}a\right)\gg\frac1{n^2}.
\end{equation}
With a similar bound for the contribution of the $y$-coordinate,
we find that the best bound we can obtain is
\begin{equation}
|\Phi_{n,m}(p)|^2\gg \frac1{n^2m^2} \gg \frac1{E^2}
\end{equation}
and this bound is sharp.

%
\section{Assumption \ref{ass:main} for the event times of a Poisson process}
\label{app:B}
The purpose of this appendix is to prove the following result. Let
$0<q<1/2$ and $1<\rho{} <2(1-q)$ be fixed throughout.
\begin{proposition} \label{prop:B:eins}
  Let $P=(E_j)_{j=1}^\infty$ be the sequence of event times for a Poisson 
process with parameter 1. There is, almost surely, a sequence 
$(\epsilon_n)_{n=1}^\infty$, $\epsilon_n\downarrow 0$ such that for each
$n$ there are four consecutive members of $P$, $\Ea<\Eb<\Ec<\Ed<
\epsilon_n^{-\rho}$, satisfying   
\begin{align} \nonumber
  \Ec-\Eb &< \epsilon_n \\
  \Ed-\Ec &> \epsilon_n^q \\
  \Eb-\Ea &> \epsilon_n^q. \nonumber
\end{align}
\end{proposition}    
Thus, assumption \ref{ass:main} is almost surely satisfied for a Poisson
process.

As a model for a Poisson process, we will let $\xi_1,\xi_2,\ldots$ be
a sequence of independent exponentially distributed random variables
with parameter 1. Then, defining
\begin{align} \nonumber
  E_1 &= \xi_1 \\
  E_2 &= \xi_1 + \xi_2 \\
   &\vdots \nonumber
\end{align}
The sequence $P=(E_j)_{j=1}^\infty$ so-formed is a Poisson process.
\begin{proposition} \label{prop:B:zwei}
  Let $\epsilon>0$. The probability that there are four consecutive members 
of $P$, $\Ea<\Eb<\Ec<\Ed<\epsilon^{-\rho}$ satisfying
\begin{align}
  \Ec-\Eb &< \epsilon \nonumber \\
  \Ed-\Ec &> \epsilon^q \\
  \Eb-\Ea &> \epsilon^q, \nonumber
\end{align}
is $1-\Ord(\epsilon^{\infty}).$
\end{proposition}

\noindent 
The notation $\Ord(\epsilon^\infty)$ refers to a quantity which goes to
zero faster than any power of $\epsilon$. One can say that
the event described in proposition \ref{prop:B:zwei} occurs with
overwhelming probability.

Let us fix $1<\rho'<\rho$, and chose $N=3M\sim\epsilon^{-\rho'}$ where
$M\in\N$.
Let us define the events ${\mathcal S}_j$, $j=0,\ldots,M-1$, by
\begin{equation}
{\mathcal S}_j=\{ \xi_{3j+1}>\epsilon^q,\; \xi_{3j+2}<\epsilon,\; \xi_{3j+3}
>\epsilon^q\}.
\end{equation}
\begin{lemma} \label{lem:B:eins}
  The events ${\mathcal S}_j$, $j=0,\ldots, M-1$ are independent, and the
probability that at least one of them occurs is $1-\Ord(\epsilon^{\infty})$.
\end{lemma}
\proof
The independence of the events ${\mathcal S_j}$ clearly follows because they
are defined on independent random variables. We first calculate the probability
of one of them.
By independence of $\xi_1,\xi_2,\xi_3$,
\begin{align} \nonumber
  \Prob({\mathcal S}_0)&=\Prob(\xi_2<\epsilon)\Prob(\xi_1>\epsilon^q)
\Prob(\xi_3>\epsilon^q) \\
&=\left( \int_0^\epsilon \rme^{-x}\,\rmd x\right)\left(\int_{\epsilon^q}^\infty
\rme^{-x}\,\rmd x\right)^2 \nonumber \\
&=(1-\rme^{-\epsilon})\left(\rme^{-\epsilon^q}\right)^2 \nonumber \\
&=\epsilon + \Ord(\epsilon^{1+q}).
\end{align}
Then, by independence of the ${\mathcal S}_j$s,
\begin{align} \nonumber
  p_1\coloneq\Prob(\mbox{at least one ${\mathcal S}_j$ occurs})&=1-\left(1-
\Prob({\mathcal S}_0)\right)^M\\
&=1-\left(1-\epsilon+\Ord(\epsilon^{1+q})\right)^M.
\end{align}
So, we have
\begin{equation}
\log(1-p_1)=-M\epsilon + \Ord(M\epsilon^{1+q}) 
\sim -\frac13\epsilon^{1-\rho'}.
\end{equation}
For $\epsilon$ sufficiently small, this yields
\begin{equation}
1-p_1\leq \exp\left(-\textstyle \frac16\epsilon^{-(\rho'-1)}\right)=
\Ord(\epsilon^\infty).
\end{equation}
\finire

The probability that the upper bound of $\epsilon^{-\rho}$ is met is given
in the following lemma
\begin{lemma} \label{lem:B:zwei}
  The probability that $E_N < \epsilon^{-\rho}$ is $1-\Ord(\epsilon^\infty)$.
\end{lemma}
\proof
Let $\alpha>0$. The probability density for $E_N$ is $\Gamma(N)^{-1}x^{N-1}
\rme^{-x}$.
So
\begin{align} \nonumber
p_2\coloneq
  \Prob(E_N < N^{1+\alpha}) &= 1-\frac1{\Gamma(N)}\int_{N^{1+\alpha}}^\infty
 x^{N-1}\rme^{-x}\,\rmd x \\
&=1-\frac{\exp(-N^{1+\alpha})}{\Gamma(N)}\int_0^\infty (x+N^{1+\alpha})^{N-1}
\rme^{-x}\,\rmd x  \nonumber \\
&=1-\exp(-N^{1+\alpha}) N^{(1+\alpha)(N-1)}\sum_{j=0}^{N-1}\frac1{\Gamma(N-j)
N^{(1+\alpha)j}},
\end{align}
expanding the binomial. Using $\displaystyle \frac1{\Gamma(N-j)}\leq
\frac{N^j}{\Gamma(N)}$ we can estimate
\begin{equation}
\sum_{j=0}^{N-1}\frac1{\Gamma(N-j)
N^{(1+\alpha)j}} \leq \frac1{\Gamma(N)}\sum_{j=0}^{N-1}\frac1{N^{\alpha j}}
\ll \frac1{\Gamma(N)},
\end{equation}
where the implied constant could depend on $\alpha$.
This leads to
\begin{align} \nonumber
  1-p_2 &\ll \frac{\exp(-N^{1+\alpha})N^{(1+\alpha)(N-1)}}{\Gamma(N)}\\
 &\sim \frac{\exp(-N^{1+\alpha}+N)N^{\alpha(N-1)}}{\sqrt{2\pi(N-1)}} \nonumber\\
 &\ll \frac1{N^\infty},
\end{align}
where Stirling's formula has been used.
This last line is $\Ord(\epsilon^\infty)$ since $N^{-1}\sim \epsilon^{\rho'}.$
Finally setting 
\begin{equation}
\alpha=\frac{\rho}{\rho'}-1>0
\end{equation}
 gives the required
estimate.
\finire
\begin{proof}[Proof of proposition \ref{prop:B:zwei}]
  We are interested in the events corresponding to lemmata \ref{lem:B:eins}
and \ref{lem:B:zwei} happening simultaneously. By the inclusion-exclusion
principle, the probability that this happens is at least $p_1+p_2-1 =
1-\Ord(\epsilon^{\infty})$.
\end{proof}
\begin{proof}
  [Proof of proposition \ref{prop:B:eins}]
Let ${\mathcal E}_n$, $n\in\N$ 
be the event that there are found four consecutive
members of $P$,  $\Ea<\Eb<\Ec<\Ed<n^{\rho}$ satisfying
\begin{align}
  \Ec-\Eb &< \frac1n \nonumber \\
  \Ed-\Ec &> \frac1{n^q} \\
  \Eb-\Ea &> \frac1{n^q}. \nonumber
\end{align}
By proposition \ref{prop:B:zwei}, 
$\Prob({\mathcal E}_n^{\rm c})\ll\frac1{n^2}$. Hence, by the Borel-Cantelli
lemma, the probability that infinitely many ${\mathcal E}_n^{\rm c}$ occur
is zero. Equivalently, only finitely many ${\mathcal E}_n^{\rm c}$ occur,
almost surely. So, almost surely, there is an infinite subsequence of 
$n\in\N$ such that ${\mathcal E}_n$ occurs.
\end{proof}

\def\rmi{{\mathrm i}}\def\Dbar{\leavevmode\lower.6ex\hbox to 0pt{\hskip-.23ex
  \accent"16\hss}D} \def\cprime{$'$}

\end{document}